\documentclass[journal,11pt,twocolumn]{IEEEtran}
\usepackage{graphicx}
\usepackage{cite}
\usepackage{array}
\usepackage{tabulary}
\newcolumntype{K}[1]{>{\centering\arraybackslash}p{#1}}

\usepackage{chemarrow}
\usepackage{color}
\usepackage{algorithmic}
\usepackage{algorithm}
\usepackage{clrscode}
\usepackage{amsmath, amssymb, mathrsfs, amsfonts, amsthm}
\usepackage{epsfig, epstopdf}
\usepackage{multirow}
\usepackage{bm}

\allowbreak
\usepackage[font=small]{caption}
\usepackage{mathdots}
\usepackage{stfloats}
\usepackage{subcaption}
\usepackage{enumitem,kantlipsum}
\usepackage{enumerate}
\usepackage{float}

\theoremstyle{theorem}
\newtheorem{lemma}{Lemma}

\newtheorem{remark}{Remark}

\newcommand{\Prob}{\mathrm{P}}

\renewcommand{\arraystretch}{1.2}
\newcommand\blfootnote[1]{%
  \begingroup
  \renewcommand\thefootnote{}\footnote{#1}%
  \addtocounter{footnote}{-1}%
  \endgroup
}

\newcommand\sbullet[1][.75]{\mathbin{\vcenter{\hbox{\scalebox{#1}{$\bullet$}}}}}
\allowdisplaybreaks

\date{}
\title{Interference Alignment Using Reaction in Molecular Interference Channels}
\author{\IEEEauthorblockN{Maryam Farahnak-Ghazani, Mahtab Mirmohseni, and Masoumeh Nasiri-Kenari
}\\
\IEEEauthorblockA{Sharif University of Technology}}

\begin{document}
\maketitle
\begin{abstract}
Co-channel interference (CCI) is a performance limiting factor in molecular communication (MC) systems with shared medium. Interference alignment (IA) is a promising scheme to mitigate CCI in traditional communication systems. Due to the signal-dependent noise in MC systems, the traditional IA schemes are less useful in MC systems. In this paper, we propose a novel IA scheme in molecular interference channels (IFCs), based on the choice of releasing/sampling times. To cancel the aligned interference signals and reduce the signal dependent noise, we use molecular reaction in the proposed IA scheme. We obtain the feasible region for the releasing/sampling times in the proposed scheme. Further, we investigate the error performance of the proposed scheme. Our results show that the proposed IA scheme using reaction improves the performance significantly.\blfootnote{This work was supported in part by the Iran National Science Foundation (INSF) Research Grant on Nano-Network Communications and in part by the Research Center of Sharif University of Technology.}
\blfootnote{The authors are with the Department of Electrical Engineering, Sharif University of Technology, Tehran, Iran (email: maryam.farahnak@ee.sharif.edu, mirmohseni@sharif.edu, mnasiri@sharif.edu).}
\end{abstract}

\section{Introduction}

In MC networks, where molecules are used as information carriers \cite{gohari2016information}, multiple transmitter-receiver pairs may operate in a shared medium \cite{dinc2017theoretical}.
In this case, if the transmitters use the same molecule type for communication and operate at the same time, there is a co-channel interference (CCI) and we have a molecular interference channel (IFC). Inter-symbol interference (ISI) and CCI are two performance limiting factors in MC systems \cite{kuran2011co}. ISI mitigating techniques in MC systems are presented in \cite{IWCIT2016type, farahnak2018medium, noel2014improving}.

The effect of CCI is considered in \cite{kuran2011co, pierobon2012intersymbol}. To avoid CCI in MC systems, the common approach is to use different molecule types for each transmitter-receiver pair, which is analogous to using different frequencies in wireless communication systems. As shown in \cite{kuran2011co}, the effect of CCI is negligible after a certain distance called molecular re-use distance and the same molecule type can be used after this distance, which is similar to the concept of frequency re-use distance in wireless communication systems. However, when the distance is less than the molecular re-use distance, CCI is not negligible, and we should use different molecule types to avoid CCI. Using different molecule types, we achieve the degrees of freedom (DoF) of $\frac{1}{2}$. However, as shown in \cite{cadambe2008interference}, a $K$-user IFC with single antenna at the transmitters and the receivers has the DoF of $\frac{K}{2}$, i.e., each user can reach half of the capacity it can reach in the absence of interference. Interference alignment (IA) is a technique to increase the DoF of a network when the distances of the transmitter-receiver pairs are less than the re-use distance, and the interference signal is comparable to the desired signal. The idea of IA is to align the signal vectors at the receivers such that the interference signals occupy the same space, while the desired signal lies in a separate space \cite{jafar2011interference}.
With the help of IA techniques, we can use the resources in the network more efficiently.

In molecular communication (MC) systems, due to the low data rate and the limitation in the number of molecule types, the usage of IA techniques in the networks can be very beneficial. 
 In this paper, we propose an IA scheme to mitigate the effect of CCI when the distance is less than the re-use distance. 
In classic communication systems, for $K$-user IFC with time/frequency varying channel coefficients, an asymptotic IA scheme is proposed to achieve the DoF of $\frac{K}{2}$ as the number of time slots or frequencies used for a super-symbol increases \cite{cadambe2008interference}. This scheme can be applied to MC systems using multiple time slots 
 or multiple molecule types 
 to form a super-symbol. However, this scheme is not practical since to achieve the DoF of $\frac{K}{2}$, the number of time slots or molecule types should be very large.
Another scheme in classic IFCs is to align interference signals by changing the channel coefficients, 
or changing the propagation delays \cite{cadambe2008interference}. To adopt this idea to MC systems, in this paper, we make use of changing the channel coefficients by choosing different releasing times at the transmitters and/or sampling times at the receivers. Based on this idea, we propose an IA scheme by the choice of releasing/sampling times in molecular IFCs. \textcolor{black}{The transmitters and receivers should be synchronized for the proposed IA scheme. The synchronization methods in MC systems are investigated in \cite{huang2020clock, 
jamali2017symbol, li2019clock}.
In this paper, to avoid complexity, we assume perfect synchronization among the transmitters and receivers. }

We also face another challenge in applying IA schemes to MC systems. Most IA schemes in classic communications, e.g., the asymptotic IA scheme,
require high signal-to-noise ratio (SNR)
\cite{jafar2011interference}. However, in MC systems, there is usually a signal dependent receiver noise \cite{gohari2016information},
which results in high noise levels for high signal levels, and hence, the classic IA schemes may not be very useful in MC systems. 
In this paper, we propose using molecular reaction \cite{farahnak2018medium} to cancel the aligned interference signals in the medium and reduce the signal dependent noise. We investigate the error performance and show that the proposed IA scheme using reaction improves the performance significantly.

Our main contributions in this paper are as follows: \\
$\sbullet$ We apply the asymptotic IA scheme in the classic communications to a $K$-user molecular IFC.\\
$\sbullet$ We propose an IA scheme based on the choice of releasing/sampling times in molecular systems for a 3-user IFC using two molecule types, and use the reaction in the proposed IA scheme to cancel the aligned interference signals in the medium and to reduce the signal dependent noise.\\
$\sbullet$ We obtain the feasible region for the releasing/sampling times in the proposed IA scheme with and without reaction. For a special case of the releasing/sampling times, we simplify the equations.\\
$\sbullet$ We investigate the error performance of the interference channel using the proposed IA scheme with and without reaction. It is seen that the IA scheme with reaction improves the performance of the system significantly.

The structure of the paper is as follows: In Section~\ref{sec:model}, we describe the system model. In Section~\ref{sec:IA}, we describe the asymptotic IA and the proposed IA schemes. In Section~\ref{sec:opt_times}, we investigate the optimum and sub-optimum sampling times and in Section~\ref{sec:error_prob}, we obtain the error probabilities of the IA schemes. The numerical results are given in Section~\ref{simulation}. Finally, in Section~\ref{conclusion} we conclude the paper.

\textbf{Notation}: Throughout the paper, vectors and matrices are shown with bold letters, $\bm{X}^\textrm{T}$ shows the transpose of vector or matrix $\bm{X}$, $\textrm{diag}(\bm{A})$ is a diagonal matrix whose entries are the elements in the vector $\bm{A}$, and $\bm{A} \equiv \bm{B}$ is equivalent to $\textrm{column-span}(\bm{A}) =\textrm{column-span}(\bm{B})$. 

\section{System Model }\label{sec:model}
We consider an IFC with $K$ molecular transmitter-receiver pairs (Fig.~\ref{fig1}). The location of the $j$-th transmitter ($\textrm{Tx}_j$) and the $i$-th receiver ($\textrm{Rx}_i$), for $i,j\in\{1,...,K\}$, are noted by $\bm{r}_j^{\textrm{Tx}}$ and $\bm{r}_i^{\textrm{Rx}}$, respectively, and the message of $\textrm{Tx}_j$ is noted by $M_j \in \{0,...,M-1\}$. The time is slotted with duration $T_\textrm{s}$.

\begin{figure}[t]
\centering
\includegraphics[scale=0.6]{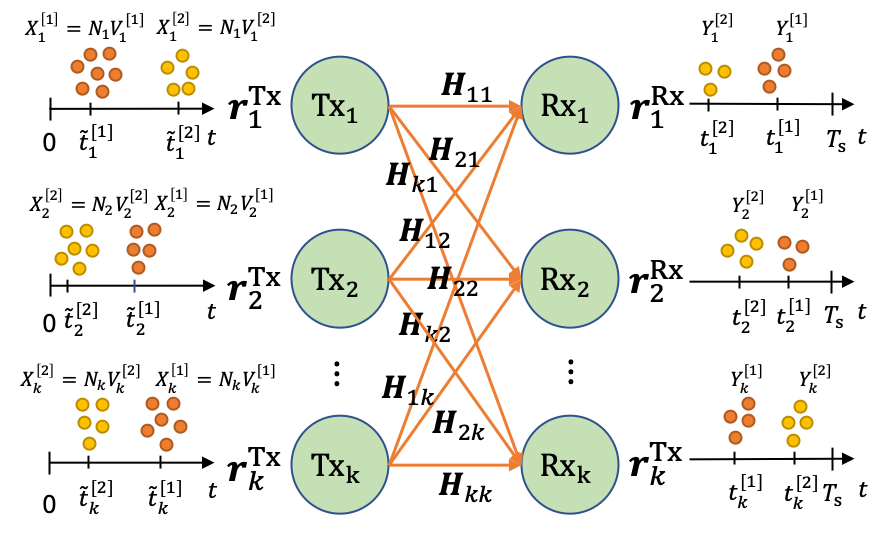}
\caption{$K$-user molecular IFC.}
\label{fig1}
\vspace{-1.5em}
\end{figure}

\textbf{Transmitter}: Each $\textrm{Tx}_j, j=1,...,K$ uses the same $L$ molecule types, and releases $X_{j}^{[l]}=N_jV_{j}^{[l]}, l=1,...,L$ molecules of type $l$ in time $\tilde{t}_j^{[l]}$, where $0<\tilde{t}_j^{[l]}< T_\textrm{s}$ and $N_j \in \{\zeta_0,...,\zeta_{M-1}\}$ is the total number of molecules of all types released from $\textrm{Tx}_j$, i.e., $\sum_{l=1}^{L}X_{j}^{[l]}=N_j$, which requires $\sum_{l=1}^{L}V_{j}^{[l]}=1$. We use concentration shift keying (CSK) modulation at the transmitters, i.e., for the message $M_j=m$, $m=0,..., M-1$, the total released molecules at $\textrm{Tx}_j$ is $N_j=\zeta_m$.
We denote the number of released molecules of all types from $\textrm{Tx}_j$ with vector $\bm{X}_j=[X_{j}^{[1]},...,X_{j}^{[L]}]^\textrm{T}=N_j\bm{V}_j$, where $\bm{V}_j=[V_{j}^{[1]},...,V_{j}^{[L]}]^\textrm{T}$ is called the beamforming vector.

\textbf{Channel}: The channel impulse response from $\textrm{Tx}_j$ to $\textrm{Rx}_i$ for the $l$-th molecule type, noted by $h_{ij}^{[l]}(t)$, is defined as the concentration of the $l$-th molecule type at $\textrm{Rx}_i$ and in time $t$ when one molecule is released from $\textrm{Tx}_j$ in time $\tilde{t}_{j}^{[l]}$. If different molecule types 
do not interfere with each other,
$h_{ij}^{[l]}(t)$ can be obtained from advection-diffusion equation for 3-D space as\cite{gohari2016information}
\begin{align}\label{Channel_impulse}
h_{ij}^{[l]}(t)=\frac{1[t>\tilde{t}_j^{[l]}]}{(4\pi D_l(t-\tilde{t}_j^{[l]}))^{\frac{3}{2}}}e^{-\frac{||\bm{r}_i^{\textrm{Rx}}-\bm{r}_j^{\textrm{Tx}}-\bm{\nu}(t-\tilde{t}_j^{[l]})||^2}{4D_l(t-\tilde{t}_j^{[l]})}}, 
\end{align}
for $i,j \in\{1,...,K\}$ and $l\in \{1,...,L\}$, where $D_l$ is the diffusion coefficient of the $l$-th molecule type and $\bm{\nu}$ is the flow velocity of the medium. For simplicity of exposition, for the analytical derivations, we assume that the ISI is negligible in the channel, which can be satisfied by using enzymes that remove the remaining molecules from the previous transmissions in the medium \cite{noel2014improving}. In Section \ref{simulation}, we investigate the error probability of the proposed IA scheme in the presence of ISI using simulation.
In the absence of ISI, the concentration of the $l$-th molecule type at $\textrm{Rx}_i$ from $\textrm{Tx}_j$ can be obtained as $C_{ij}^{[l]}(t)=X_{i}^{[l]}h_{ij}^{[l]}(t)$. 

\textbf{Receiver}: Each receiver $\textrm{Rx}_i, i=1,....,K$ is assumed to be transparent with volume $V_\textrm{R}$ (radius $r_\textrm{R}$), and counts the number of molecules that enters its volume. We assume that $\textrm{Rx}_i$ uses one-sample decoder with sampling time $t_{i}^{[l]}$ for the $l$-th molecule type, where ${\tilde{t}_{j}^{[l]}<t_{i}^{[l]}<T_\textrm{s}}$, $i,j \in \{1,...,K\}, l \in \{1,...,L\}$.\footnote{The one sample decoder at the receiver is commonly used in MC systems, as in \cite{noel2014improving, akdeniz2020equilibrium}.} The releasing and sampling times at each transmitter and receiver are shown in Fig.~\ref{fig1} for two molecule types, i.e., $L=2$. Assuming uniform concentration in the receiver volume, which is valid if the distance of the transmitter and the receiver is sufficiently large compared to the receiver radius \cite{noel2014diffusive}, the mean number of counted molecules at $\textrm{Rx}_i$ from $\textrm{Tx}_j$ is obtained as $\mu_{ij}^{[l]}=C_{ij}^{[l]}(t_{i}^{[l]})V_\textrm{R}=X_{i}^{[l]}h_{ij}^{[l]}(t_{i}^{[l]})V_\textrm{R}$. We denote $H_{ij}^{[l]}=h_{ij}^{[l]}(t_{i}^{[l]})V_\textrm{R}$. Hence, the mean number of counted molecules of type $l$ at $\textrm{Rx}_i$ transmitted from all transmitters is $\mu_i^{[l]}=\sum_{j=1}^{K}\mu_{ij}^{[l]}=\sum_{j=1}^{K}X_j^{[l]} H_{ij}^{[l]}$.
We denote the vector of mean number of counted molecules of all types at $\textrm{Rx}_i$ with vector $\bm{\mu}_i=[\mu_i^{[1]},...,\mu_i^{[L]}]^\textrm{T}$. Hence, $\bm{\mu}_i=\sum_{j=1}^{K}\bm{H}_{ij}\bm{X}_j$, where
\small
\begin{equation}
\begin{aligned}
\bm{H}_{ij}=\textrm{diag}(H_{ij}^{[1]},H_{ij}^{[2]},...,H_{ij}^{[L]}),\qquad i,j \in\{1,...,K\}.
\end{aligned}
\end{equation}
\normalsize
We assume that the channel state information (CSI) is known at the transmitters and receivers. 
\textcolor{black}{Further, we assume that the transmitters and receivers are synchronized.}

\textbf{Noise}: Let $Y_i^{[l]}$ be the number of counted molecules of type $l$ at $\textrm{Rx}_i$ transmitted from all transmitters. Assuming the counting noise at the receivers and an environment noise with mean $\mu_\textrm{n}^{[l]}$ for the $l$-th molecule type, $Y_i^{[l]}$ has a Poisson distribution conditioned on the transmitter messages with mean $\mu_i^{[l]}+\mu_\textrm{n}^{[l]}$ \cite{gohari2016information}. So, we have a signal dependent noise at the receivers. For simplicity of the analysis, we assume $\mu_\textrm{n}^{[l]}=0$. However, in Section \ref{simulation}, we investigate the effect of the environment noise.

\section{Interference Alignment}\label{sec:IA}
In this section, we first apply the asymptotic IA scheme of \cite{cadambe2008interference} to $K$-user molecular IFC, where we use different molecule types or time slots for sending a super-symbol and increase their number to reach the DoF of $\frac{K}{2}$, asymptotically. However, we should note that increasing the number of molecule types or the number of time slots is not practical. A practical solution is to partition the $K$-user pairs to clusters with $3$-user pairs and use this scheme with $3$ molecule types in each cluster, as in \cite{loch2014practical}, where we reach the DoF of $\frac{4}{3}$. Further, this scheme is only useful in systems with high SNRs. Since in MC systems, we face a signal dependent receiver noise, which makes the SNR low, the traditional IA schemes are less useful. To overcome this issue, in the second part, we propose a non-asymptotic IA scheme by the choice of releasing/sampling times for 3-user IFC using 2 molecule types, which has a DoF of $\frac{3}{2}$. We use reaction in the proposed IA scheme to reduce the signal dependent noise. We should note that using the asymptotic IA scheme for 3 user pairs, we can't reach the DoF of $\frac{3}{2}$ with limited number of molecule types (for 3 molecule types, we can reach the DoF of $\frac{4}{3}$) and we should increase the number of molecules types to reach the DoF of $\frac{3}{2}$, asymptotically.

\textbf{Asymptotic IA scheme}: We consider a super-symbol at each transmitter $\textrm{Tx}_i$ consisting of $n_{\textrm{s},i}$ symbols and send it using different variations of the channel to reach the DoF of $\frac{K}{2}$, asymptotically. 
To have the varying channel gains, we can use either different molecule types in one time slot, or one molecule type in multiple time slots with different releasing/sampling times in each time slot. Here, we assume $L$ different molecule types in one time slot. Hence, according to the system model, the number of released molecules of all types from $\textrm{Tx}_j$ is obtained from $\bm{X}_j=N_j\bm{V}_j$. Using the method in \cite{cadambe2008interference}, we assume $n_{\textrm{s},1}=(n+1)^N$, $n_{\textrm{s},i}=n^N, i=2,...,K$, and $L=(n+1)^N+n^N$, in which $N=(K-1)(K-2)-1$ and $n$ is an auxiliary variable, which is going to be increased in order to achieve the DoF of $\frac{K}{2}$. The beamforming vectors at the transmitters are chosen as
\small
\begin{align}
\nonumber
\bm{V}_1&=\bigg\{\prod_{\substack{g,q \in \{2,...,K\},\\ m \neq g, (g,q)\neq (2,3)}}(\bm{T}_{gq})^{\alpha_{gq}}\bm{w}: \forall  \alpha_{gq} \in\{0,...,n\}\bigg\}, \\
\bm{V}_j&=\bm{S}_j\bm{B}, \qquad j=2,...,K,
\end{align}
\normalsize
where $\bm{w}$ is a $L \times 1$ vector as $\bm{w}=[1,...,1]^\textrm{T}$, and
\small
\begin{equation}
\begin{aligned}
\notag
&\bm{S}_j=(\bm{H}_{1j})^{-1}\bm{H}_{13}(\bm{H}_{23})^{-1}\bm{H}_{21}, \qquad j=2,...,K,\\
&\bm{T}_{ij}=(\bm{H}_{i1})^{-1}\bm{H}_{ij}\bm{S}_j,\qquad i,j=2,...,K, \quad i\neq j,\\
&\bm{B}=\bigg\{\prod_{\substack{g,q \in \{2,..K\},\\ g \neq q, (g,q)\neq (2,3)}}(\bm{T}_{gq})^{\alpha_{gq}}\bm{w}: \forall  \alpha_{gq} \in\{0,...,n-1\}\bigg\}.
\end{aligned}
\end{equation}
\normalsize
The DoF is equal to $\frac{(n+1)^N+(K-1)n^N}{(n+1)^N+n^N}$ and by increasing $n$, we can reach the DoF of $\frac{K}{2}$. 

\textbf{Proposed non-asymptotic IA scheme by choice of releasing/sampling times}: We assume $K=3$ and $L=2$. 
For the 3-user IFC, the vector of mean number of counted molecules at $\textrm{Rx}_i$,  $i \in \{1,2,3\}$, is
\begin{align}\label{signals}
\nonumber
\bm{\mu}_i&=\sum_{j=1}^{3}\bm{H}_{ij}X_j=\bm{H}_{i1}\bm{X}_1+\bm{H}_{i2}\bm{X}_2+\bm{H}_{i3}\bm{X}_3\\
&= \bm{H}_{i1}\bm{V}_1N_1+ \bm{H}_{i2}\bm{V}_2N_2+ \bm{H}_{i3}\bm{V}_3N_3.
\end{align}
For IA, we should have
\begin{equation}\label{cond12}
\begin{aligned}
\textrm{Rx}_1: \quad \bm{H}_{12}\bm{V}_2\equiv\bm{H}_{13}\bm{V}_3\not \equiv\bm{H}_{11}\bm{V}_1,\\
\textrm{Rx}_2: \quad \bm{H}_{21}\bm{V}_1\equiv\bm{H}_{2 3}\bm{V}_3\not\equiv\bm{H}_{22}\bm{V}_2,\\
\textrm{Rx}_3: \quad \bm{H}_{31}\bm{V}_1\equiv\bm{H}_{32}\bm{V}_2\not\equiv\bm{H}_{33}\bm{V}_3.
\end{aligned}
\end{equation}
From \eqref{cond12}, if the channel matrices are such that $T=\bm{H}_{13}^{-1}\bm{H}_{12}\bm{H}_{21}^{-1} \bm{H}_{23}\bm{H}_{32}^{-1}\bm{H}_{31}\equiv I$, where $I$ is an identity matrix of order 2, or equivalently,
\begin{equation}\label{cond1}
\begin{aligned}
\frac{H_{13}^{[2]}}{H_{13}^{[1]}}.\frac{H_{12}^{[1]}}{H_{12}^{[2]}}.\frac{H_{21}^{[2]}}{H_{21}^{[1]}}.\frac{H_{23}^{[1]}}{H_{23}^{[2]}}.\frac{H_{32}^{[2]}}{H_{32}^{[1]}}.\frac{H_{31}^{[1]}}{H_{31}^{[2]}}=1,
\end{aligned}
\end{equation}
we can choose the following beamforming vectors that satisfy \eqref{cond12}:
\begin{align}\label{beam_vec}
\nonumber
\bm{V}_1&\equiv[1,1]^{\textrm{T}},\quad \bm{V}_2\equiv\bm{H}_{32}^{-1} \bm{H}_{31} \bm{V}_1=\left[\frac{H_{31}^{[1]}}{H_{32}^{[1]}}, \frac{H_{31}^{[2]}}{H_{32}^{[2]}}\right]^{\textrm{T}},\\
\bm{V}_3&\equiv\bm{H}_{23}^{-1} \bm{H}_{21} \bm{V}_1=\left[\frac{H_{21}^{[1]}}{H_{23}^{[1]}}, \frac{H_{21}^{[2]}}{H_{23}^{[2]}}\right]^{\textrm{T}}.
\end{align}
We should normalize the beamforming vectors such that we have a unit norm, i.e., $|\bm{V}_i|_1=1$. Further, \eqref{cond12} requires the following conditions:
\begin{equation}\label{cond2}
\begin{aligned}
 \bm{H}_{12}\bm{V}_2\not\equiv\bm{H}_{11}\bm{V}_1, ~ \bm{H}_{21}\bm{V}_1\not\equiv\bm{H}_{22}\bm{V}_2, ~ \bm{H}_{31}\bm{V}_1\not\equiv\bm{H}_{33}\bm{V}_3.
\end{aligned}
\end{equation}
We call the condition in \eqref{cond1} as the IA condition and the conditions in \eqref{cond2} as the independency conditions. Our goal is to change the channel coefficients such that the IA and independency conditions hold, thus we can align the interference signals at the receivers by choosing the beamforming vectors at the transmitters as \eqref{beam_vec}. 
\newline
Using the beamforming vectors in \eqref{beam_vec}, the mean signal vectors at the receivers are obtained from \eqref{signals} as
\small
\begin{equation}\label{received_signals}
\begin{aligned}
\textrm{Rx}_1: \bm{\mu}_1&=\frac{N_1}{2}\left[\begin{matrix}
H_{11}^{[1]}\\H_{11}^{[2]}
\end{matrix}\right]
+\frac{N_2}{A}\left[\begin{matrix}
H_{12}^{[1]}\frac{H_{31}^{[1]}}{H_{32}^{[1]}}\\H_{12}^{[2]}\frac{H_{31}^{[2]}}{H_{32}^{[2]}}
\end{matrix}\right]+\frac{N_3}{B}\left[\begin{matrix}
H_{13}^{[1]}\frac{H_{21}^{[1]}}{H_{23}^{[1]}}\\H_{13}^{[2]}\frac{H_{21}^{[2]}}{H_{23}^{[2]}}
\end{matrix}\right],\\
\textrm{Rx}_2: \bm{\mu}_2&=\frac{N_1}{2}\left[\begin{matrix}
H_{21}^{[1]}\\H_{21}^{[2]}
\end{matrix}\right]
+\frac{N_2}{A}\left[\begin{matrix}
H_{22}^{[1]}\frac{H_{31}^{[1]}}{H_{32}^{[1]}}\\H_{22}^{[2]}\frac{H_{31}^{[2]}}{H_{32}^{[2]}}
\end{matrix}\right]
+\frac{N_3}{B} \left[\begin{matrix}
H_{21}^{[1]}\\H_{21}^{[2]}
\end{matrix}\right],\\
\textrm{Rx}_3: \bm{\mu}_3&=\frac{N_1}{2}\left[\begin{matrix}
H_{31}^{[1]}\\H_{31}^{[2]}
\end{matrix}\right]
+\frac{N_2}{A}\left[\begin{matrix}
H_{31}^{[1]}\\H_{31}^{[2]}
\end{matrix}\right]
+\frac{N_3}{B}\left[\begin{matrix}
H_{33}^{[1]}\frac{H_{21}^{[1]}}{H_{23}^{[1]}}\\H_{33}^{[2]}\frac{H_{21}^{[2]}}{H_{23}^{[2]}}
\end{matrix}\right],
\end{aligned}
\end{equation}
\normalsize
where $A=\frac{H_{31}^{[1]}}{H_{32}^{[1]}}+\frac{H_{31}^{[2]}}{H_{32}^{[2]}}$ and $B=\frac{H_{21}^{[1]}}{H_{23}^{[1]}}+\frac{H_{21}^{[2]}}{H_{23}^{[2]}}$. From \eqref{received_signals}, it can be seen that the interference signals at Rx$_2$ and Rx$_3$ are aligned without resorting to \eqref{cond1}.
However, to align the interference signals at Rx$_1$, \eqref{cond1} is needed to hold. For this purpose, we change the channel coefficients by the choice of the releasing times at the transmitters, $\tilde{t}_{j}^{[l]}$, and/or the sampling times at the receivers, $t_i^{[l]}$, to meet the conditions in \eqref{cond1} and \eqref{cond2}.\footnote{In MC systems, releasing time is usually chosen as zero, i.e., the molecules are released at the beginning of the time slot, and the sampling time is usually chosen such that the concentration of the molecules is maximized at the time \cite{noel2014improving}. However, it can be shown that if we use these times for each transmitter-receiver pair, the desired signal and the interference signals will be all aligned and the independency condition cannot be met. }
In the following, we simplify the alignment and independency equations in \eqref{cond1} and \eqref{cond2} to obtain a feasible region for $t_i^{[l]}$s and $\tilde{t}_{j}^{[l]}$s. 
\newline
\begin{figure*}
 \begin{align}\label{cond1_simplified}
&r_{13}^2.f(\Delta t_{13}^{[1]},\Delta t_{13}^{[2]})-r_{12}^2.f(\Delta t_{12}^{[1]},\Delta t_{12}^{[2]})+r_{21}^2.f(\Delta t_{21}^{[1]},\Delta t_{21}^{[2]})-r_{23}^2.f(\Delta t_{23}^{[1]},\Delta t_{23}^{[2]})\\\nonumber
&\quad +r_{32}^2.f(\Delta t_{32}^{[1]},\Delta t_{32}^{[2]})-r_{31}^2.f(\Delta t_{31}^{[1]},\Delta t_{31}^{[2]})+6\ln(\frac{\Delta t_{13}^{[1]}}{\Delta t_{13}^{[2]}}.\frac{\Delta t_{12}^{[2]}}{\Delta t_{12}^{[1]}}.\frac{\Delta t_{21}^{[1]}}{\Delta t_{21}^{[2]}}.\frac{\Delta t_{23}^{[2]}}{\Delta t_{23}^{[1]}}.\frac{\Delta t_{32}^{[1]}}{\Delta t_{32}^{[2]}}.\frac{\Delta t_{31}^{[2]}}{\Delta t_{31}^{[1]}})=0.
\end{align}
\hrulefill
\vspace{-1.8em}
\end{figure*}
\begin{figure*}
\small
\begin{align}\label{cond2_simplified}
\nonumber
\textrm{Rx}_1: &r_{11}^2.f(\Delta t_{11}^{[1]},\Delta t_{11}^{[2]})-r_{12}^2.f(\Delta t_{12}^{[1]},\Delta t_{12}^{[2]})+r_{32}^2.f(\Delta t_{32}^{[1]},\Delta t_{32}^{[2]})-r_{31}^2.f(\Delta t_{31}^{[1]},\Delta t_{31}^{[2]})+6\ln(\frac{\Delta t_{11}^{[1]}}{\Delta t_{11}^{[2]}}.\frac{\Delta t_{12}^{[2]}}{\Delta t_{12}^{[1]}}.\frac{\Delta t_{32}^{[1]}}{\Delta t_{32}^{[2]}}.\frac{\Delta t_{31}^{[2]}}{\Delta t_{31}^{[1]}})\neq 0,\\\nonumber
\textrm{Rx}_2: &r_{22}^2.f(\Delta t_{22}^{[1]},\Delta t_{22}^{[2]})-r_{21}^2.f(\Delta t_{21}^{[1]},\Delta t_{21}^{[2]})+r_{31}^2.f(\Delta t_{31}^{[1]},\Delta t_{31}^{[2]})-r_{32}^2.f(\Delta t_{32}^{[1]},\Delta t_{32}^{[2]})+6\ln(\frac{\Delta t_{22}^{[1]}}{\Delta t_{22}^{[2]}}.\frac{\Delta t_{21}^{[2]}}{\Delta t_{21}^{[1]}}.\frac{\Delta t_{31}^{[1]}}{\Delta t_{31}^{[2]}}.\frac{\Delta t_{32}^{[2]}}{\Delta t_{32}^{[1]}})\neq 0,\\
\textrm{Rx}_3: &r_{33}^2.f(\Delta t_{33}^{[1]},\Delta t_{33}^{[2]})-r_{31}^2.f(\Delta t_{31}^{[1]},\Delta t_{31}^{[2]})+r_{21}^2.f(\Delta t_{21}^{[1]},\Delta t_{21}^{[2]})-r_{23}^2.f(\Delta t_{23}^{[1]},\Delta t_{23}^{[2]})+6\ln(\frac{\Delta t_{33}^{[1]}}{\Delta t_{33}^{[2]}}.\frac{\Delta t_{31}^{[2]}}{\Delta t_{31}^{[1]}}.\frac{\Delta t_{21}^{[1]}}{\Delta t_{21}^{[2]}}.\frac{\Delta t_{23}^{[2]}}{\Delta t_{23}^{[1]}})\neq 0.
\end{align}
\normalsize
\vspace{-1.8em}
\hrulefill
\end{figure*}
Using \eqref{Channel_impulse}, the IA condition in \eqref{cond1} simplifies to \eqref{cond1_simplified}
and the independency conditions in \eqref{cond2} simplify to \eqref{cond2_simplified},
for $\tilde{t}_{j}^{[l]}<t_{i}^{[l]}<T_\textrm{s}, i,j \in \{1,2,3\}$, $l\in\{1,2\}$, where $f(\Delta t_{ij}^{[1]},\Delta t_{ij}^{[2]})=\frac{1}{D_1\Delta t_{ij}^{[1]}}-\frac{1}{D_2\Delta t_{ij}^{[2]}}$, $\Delta t_{ij}^{[l]}=t_{i}^{[l]}-\tilde{t}_{j}^{[l]}$, and $r_{ij}^2=||\bm{r}_i^{\textrm{Rx}}-\bm{r}_j^{\textrm{Tx}}||^2, i,j \in \{1,2,3\}$, $l \in \{1,2\}$. The conditions in \eqref{cond1_simplified} and \eqref{cond2_simplified} form a feasible region for $t_{i}^{[l]}$s and $\tilde{t}_j^{[l]}$s noted by the set $\mathcal{T}_\textrm{nr}$, where ``nr" stands for ``no reaction". In Lemma~\ref{lemma1}, we simplify these conditions for a special case of the releasing and sampling times.

\begin{remark} It can be seen that the conditions do not depend on the flow velocity of the medium. 
\end{remark}

\begin{lemma} \label{lemma1} 
For $t_{i}^{[1]}=t_{i}^{[2]}=t_{i}, \tilde{t}_{j}^{[1]}=\tilde{t}_{j}^{[2]}=\tilde{t}_{j}$, $i,j \in \{1,2,3\}$, the IA condition reduces to:
\begin{align}\label{cond1_simplified2}
\nonumber
&\frac{r_{13}^2}{t_{1}-\tilde{t}_{3}}-\frac{r_{12}^2}{t_{1}-\tilde{t}_{2}}+\frac{r_{21}^2}{t_{2}-\tilde{t}_{1}}\\
&\quad -\frac{r_{23}^2}{t_{2}-\tilde{t}_{3}}+\frac{r_{32}^2}{t_{3}-\tilde{t}_{2}}-\frac{r_{31}^2}{t_{3}-\tilde{t}_{1}}=0,
\end{align}
and the independency conditions reduce to
\small
\begin{equation}\label{cond2_simplified2}
\begin{aligned}
&\textrm{Rx}_1:\frac{r_{11}^2}{t_{1}-\tilde{t}_{1}}-\frac{r_{12}^2}{t_{1}-\tilde{t}_{2}}+\frac{r_{32}^2}{t_{3}-\tilde{t}_{2}}-\frac{r_{31}^2}{t_{3}-\tilde{t}_{1}}\neq 0,\\
&\textrm{Rx}_2:\frac{r_{22}^2}{t_{2}-\tilde{t}_{2}}-\frac{r_{21}^2}{t_{2}-\tilde{t}_{1}}+\frac{r_{31}^2}{t_{3}-\tilde{t}_{1}}-\frac{r_{32}^2}{t_{3}-\tilde{t}_{2}}\neq 0,\\
&\textrm{Rx}_3:\frac{r_{33}^2}{t_{3}-\tilde{t}_{3}}-\frac{r_{31}^2}{t_{3}-\tilde{t}_{1}}+\frac{r_{21}^2}{t_{2}-\tilde{t}_{1}}-\frac{r_{23}^2}{t_{2}-\tilde{t}_{3}}\neq 0.
\end{aligned}
\end{equation}
\normalsize
The conditions in \eqref{cond1_simplified2} and \eqref{cond2_simplified2} form a feasible region for $t_i$s and $\tilde{t}_j$s, noted by the set $\mathcal{T}_\textrm{nr,s}$, where ``s" stands for ``special case". 
\end{lemma}
\begin{proof}
The proof is straightforward from \eqref{cond1_simplified} and \eqref{cond2_simplified}.
\end{proof}
\textbf{Using Reaction in the IA scheme}: To cancel the aligned interference in the proposed IA scheme and reduce the signal dependent noise, we exploit molecular reaction. Consider the following reactions among the molecules of type 1 and 2 ($\textrm{M}_1$ and $\textrm{M}_2$) at $\textrm{Rx}_i$:
\begin{equation}\label{reaction}
\begin{aligned}
&\textrm{Rx}_i: c_i \textrm{M}_1+ \textrm{M}_2+\textrm{E}_i \rightarrow \textrm{products},
\end{aligned}
\end{equation}
for $i \in \{1,2,3\}$, where $\textrm{E}_i$ is an enzyme released from $\textrm{Rx}_i$, which enables $\textrm{M}_1$ and $\textrm{M}_2$ to react with coefficient $c_i$ and with high reaction rate. The products are not detectable by the receivers. \textcolor{black}{Further, we assume that the sampling times of different molecule types at the each receiver are the same, i.e., $t_i^{[1]}=t_i^{[2]}=t_i$, and enzymes are released at the sampling time $t_i$ from each receiver.} According to \eqref{reaction}, to cancel the interference signals at the receivers, we require that the ratio of the interference signal of $\textrm{M}_1$ to the interference signal of $\textrm{M}_2$ at $\textrm{Rx}_i$ be equals to $c_i$. Hence, using the signal vectors in \eqref{received_signals}, we should have
\small
\begin{equation}\label{IA_reaction}
\begin{aligned}
&\textrm{Rx}_1: \frac{H_{12}^{[1]}\frac{H_{31}^{[1]}}{H_{32}^{[1]}}}{H_{12}^{[2]}\frac{H_{31}^{[2]}}{H_{32}^{[2]}}}=c_1, ~ \textrm{Rx}_2: \frac{H_{21}^{[1]}}{H_{21}^{[2]}}=c_2,~ \textrm{Rx}_3: \frac{H_{31}^{[1]}}{H_{31}^{[2]}}=c_3.
\end{aligned}
\end{equation}
\normalsize
We call the conditions in \eqref{IA_reaction} as the reaction conditions. Hence, if it is feasible to choose the releasing/sampling times such that along with the IA and independency conditions in \eqref{cond1} and \eqref{cond2} the reaction conditions also hold, the aligned interference signals at the receivers will be canceled  using reaction. 
Using \eqref{Channel_impulse}, the reaction conditions in \eqref{IA_reaction} simplifies to \eqref{IA_reaction2},
where $f(.,.)$ is defined in \eqref{cond1_simplified}. \textcolor{black}{Note that in the IA, independency, and reaction conditions in \eqref{cond1_simplified}, \eqref{cond2_simplified}, and \eqref{IA_reaction2}, we should further have $t_i^{[1]}=t_i^{[2]}=t_i$.}
\begin{figure*}
\small
\begin{align}\label{IA_reaction2}
\nonumber
\textrm{Rx}_1: &r_{12}^2. f(\Delta t_{12}^{[1]},\Delta t_{12}^{[2]})+r_{31}^2. f(\Delta t_{31}^{[1]},\Delta t_{31}^{[2]})-r_{32}^2. f(\Delta t_{32}^{[1]},\Delta t_{32}^{[2]})+6\ln(\frac{D_1}{D_2}\frac{\Delta t_{12}^{[1]}}{\Delta t_{12}^{[2]}}\frac{\Delta t_{31}^{[1]}}{\Delta t_{31}^{[2]}}\frac{\Delta t_{32}^{[2]}}{\Delta t_{32}^{[1]}})\\\nonumber
&+||\bm{\nu}||^2(\frac{\Delta t_{12}^{[1]}+\Delta t_{31}^{[1]}-\Delta t_{32}^{[1]}}{D_1}-\frac{\Delta t_{12}^{[2]}+\Delta t_{31}^{[2]}-\Delta t_{32}^{[2]}}{D_2})-2(\bm{r}_1^{\textrm{Rx}}-\bm{r}_1^{\textrm{Tx}}).\bm{\nu}(\frac{1}{D_1}-\frac{1}{D_2})=-4\ln{c_1},\\\nonumber
\textrm{Rx}_2: &r_{21}^2. f(\Delta t_{21}^{[1]},\Delta t_{21}^{[2]})+6\ln(\frac{D_1}{D_2}\frac{\Delta t_{21}^{[1]}}{\Delta t_{21}^{[2]}})+||\bm{\nu}||^2(\frac{\Delta t_{21}^{[1]}}{D_1}-\frac{\Delta t_{21}^{[2]}}{D_2})-2(\bm{r}_2^{\textrm{Rx}}-\bm{r}_1^{\textrm{Tx}}).\bm{\nu}(\frac{1}{D_1}-\frac{1}{D_2})=-4\ln{c_2},\\
\textrm{Rx}_3: &r_{31}^2. f(\Delta t_{31}^{[1]},\Delta t_{31}^{[2]})+6\ln(\frac{D_1}{D_2}\frac{\Delta t_{31}^{[1]}}{\Delta t_{31}^{[2]}})+||\bm{\nu}||^2(\frac{\Delta t_{31}^{[1]}}{D_1}-\frac{\Delta t_{31}^{[2]}}{D_2})-2(\bm{r}_3^{\textrm{Rx}}-\bm{r}_1^{\textrm{Tx}}).\bm{\nu}(\frac{1}{D_1}-\frac{1}{D_2})=-4\ln{c_3}.
\end{align}
\normalsize
\hrulefill
\vspace{-1.5em}
\end{figure*}

Two approaches are possible: (i) We may choose the releasing and/or sampling times such that they satisfy the IA and independency conditions in \eqref{cond1_simplified} and \eqref{cond2_simplified}, and then obtain $c_1, c_2,$ and $c_3$ from \eqref{IA_reaction2}. However, we may not find reactions among $\textrm{M}_1$ and $\textrm{M}_2$ with these coefficients. (ii) We may choose $c_1, c_2$, and $c_3$ arbitrarily and then choose the releasing and/or sampling times such that they satisfy the IA, independency, and reaction conditions in \eqref{cond1_simplified}, \eqref{cond2_simplified}, and \eqref{IA_reaction2}. These conditions form a feasible region for $t_{i}^{[l]}$s and $\tilde{t}_j^{[l]}$s, noted by the set $\mathcal{T}_\textrm{r}$, where ``r" stands for ``reaction". We take the second approach and for a special case, simplify the IA, independency, and reaction conditions in the following Lemma. 

\begin{lemma} \label{lemma2}
For $t_{i}^{[1]}=t_{i}^{[2]}=t_{i}, \tilde{t}_{j}^{[1]}=\tilde{t}_{j}^{[2]}=\tilde{t}_{j}$, $i,j \in \{1,2,3\}$, $\bm{\nu}=0$, and $c_1=c_2=c_3=c$, the IA and reaction conditions reduce to
\small
\begin{equation}\label{IA_reaction7}
\begin{aligned}
\Delta \tilde{t}_{12}(1-\frac{r_{12}^2}{r_{32}^2})-\Delta \tilde{t}_{13}(1-\frac{r_{13}^2}{r_{23}^2})=s.(\frac{r_{12}^2}{r_{32}^2}.r_{31}^2-\frac{r_{13}^2}{r_{23}^2}.r_{21}^2),
\end{aligned}
\end{equation}
\normalsize
and the independency conditions reduce to
\small
\begin{align}\label{IA_reaction8}
\nonumber
&\Delta \tilde{t}_{12}\neq s.\frac{r_{22}^2.r_{31}^2-r_{21}^2.r_{32}^2}{r_{32}^2-r_{22}^2},\\
&\Delta \tilde{t}_{13}\notin \big\{s.\frac{r_{13}^2.r_{21}^2-r_{11}^2.r_{23}^2}{r_{23}^2-r_{13}^2}, s.\frac{r_{33}^2.r_{21}^2-r_{31}^2.r_{23}^2}{r_{23}^2-r_{33}^2}\big\},
\end{align}
\normalsize
where $\Delta \tilde{t}_{1j}=\tilde{t}_{1}-\tilde{t}_{j}, j=2,3$, $s=-\frac{\frac{1}{D_1}-\frac{1}{D_2}}{4\ln{c}-6\ln{\frac{D_2}{D_1}}}$. 
Further, it is needed that $\Delta t_{ij}>0, i,j \in\{1,2,3\}$, which results in
$s>0$ and
\begin{align}\label{IA_reaction10}
\nonumber
& \begin{cases}
\Delta \tilde{t}_{13}>-s\min\{r_{21}^2,r_{31}^2\}, \qquad&\textrm{if}~ r_{13}>r_{23}\\\nonumber
-s\min\{r_{21}^2,r_{31}^2\}<\Delta \tilde{t}_{13}<-\frac{s.r_{21}^2}{1-\frac{r_{23}^2}{r_{13}^2}}, ~&\textrm{if}~ r_{13}<r_{23} \end{cases}.\\
& \Delta \tilde{t}_{12}>-s\min\{r_{21}^2,r_{31}^2\}.
\end{align}
The conditions in \eqref{IA_reaction7}-\eqref{IA_reaction10} form a feasible region for $\Delta \tilde{t}_{12}$ and $\Delta \tilde{t}_{13}$, noted by $\mathcal{T}_\textrm{r,s}$. After choosing $\Delta \tilde{t}_{12}$ and $\Delta \tilde{t}_{13}$ in the feasible region, we choose one of the releasing times arbitrarily and obtain the other two releasing times from the values of $\Delta \tilde{t}_{12}$ and $\Delta \tilde{t}_{13}$. Then, $t_{1}$, $t_2$, and $t_3$ can be obtained from $\tilde{t}_{1}$ and $\tilde{t}_3$ as follows:
\begin{align}\label{IA_reaction9}
\nonumber
&t_{1}=\frac{r_{13}^2}{r_{23}^2}\tilde{t}_{1}+(1-\frac{r_{13}^2}{r_{23}^2})\tilde{t}_{3}+s.\frac{r_{13}^2}{r_{23}^2}.r_{21}^2,\\
&t_{2}=\tilde{t}_{1}+s.r_{21}^2,\quad t_{3}=\tilde{t}_{1}+s.r_{31}^2.
\end{align}
\end{lemma}

\begin{proof}
The proof is provided in Appendix \ref{apendixa}.
\end{proof}
\begin{remark}
Without loss of generality we assume $D_1>D_2$. Hence, for the spacial case of the times in Lemma~\ref{lemma2}, we should choose molecules to have reaction with $c>(\frac{D_2}{D_1})^{\frac{3}{2}}$ to garantee $s>0$.
\end{remark}

\section{Optimum Sampling and Releasing Times}\label{sec:opt_times}
 To obtain the optimum sampling and releasing times, we minimize the total error probability of the system over the feasible region of the times. Using the error probability at $\textrm{Rx}_i$, noted by $P_{\textrm{e},i}$, $i \in \{1,2,3\}$, the total error probability of the system is 
$P_\textrm{e}=\frac{1}{3}(P_{\textrm{e},1}+P_{\textrm{e},2}+P_{\textrm{e},3}).$

\textbf{IA without reaction}: For the general case of the sampling and releasing times, the times are obtained from the following optimization problem:
\begin{align}\label{times_anr_gen}
\bm{t}_\textrm{nr}=\arg \min_{\bm{t} \in \mathcal{T}_\textrm{nr}} P_\textrm{e,nr},
\end{align}
where $\bm{t}=[\tilde{t}_{1}^{[1]}, \tilde{t}_{1}^{[2]}, \tilde{t}_{2}^{[1]}, \tilde{t}_{2}^{[2]}, \tilde{t}_{ 3,}^{[1]}, \tilde{t}_{3}^{[2]}, t_{1}^{[1]}, t_{1}^{[2]}, t_{2}^{[1]}, t_{2}^{[2]},t_{3}^{[1]}$, $t_{3}^{[2]}]$,
$P_\textrm{e,nr}$ is the total error probability in the IA scheme without reaction, and as defined in Section \ref{sec:IA}, $\mathcal{T}_\textrm{nr}$ is the feasible region of the times in the IA with no reaction. For the special case of the times when $t_i^{[l]}=t_i$ and $\tilde{t}_j^{[l]}=\tilde{t}_j$, $i,j \in \{1,2,3\}$, we should solve
\begin{align}\label{times_anr_spec}
\bm{t}_\textrm{nr,s}=\arg \min_{\bm{t}_\textrm{s} \in \mathcal{T}_\textrm{nr,s}} P_\textrm{e,nr},
\end{align}
where $\bm{t}_\textrm{s}=[\tilde{t}_{1}, \tilde{t}_{2}, \tilde{t}_{3},t_{1}$, $t_{2},t_{3}]$
and $\mathcal{T}_\textrm{nr,s}$ is characterized in Lemma~\ref{lemma1}. 

\textbf{IA with reaction}: Similar to the no reaction case, for the general case of the times, we have
\begin{align}\label{times_ar_gen}
\bm{t}_\textrm{r}=\arg \min_{\bm{t} \in \mathcal{T}_\textrm{r}} P_\textrm{e,r},
\end{align}
where
$P_\textrm{e,r}$ is the total error probability in the IA scheme with reaction, and as defined in Section \ref{sec:IA}, $\mathcal{T}_\textrm{r}$ is the feasible region of the times in the IA with reaction. For the special case of the times, according to Lemma~\ref{lemma2}, the feasible region is obtained for $\Delta \tilde{t}_{12}$ and $\Delta \tilde{t}_{13}$. Hence, we should solve the following optimization problem 
\begin{equation}\label{times_ar_spec}
\begin{aligned}
&[\Delta \tilde{t}_{12, \textrm{r,s}}, \Delta \tilde{t}_{13,\textrm{r,s}}]=\arg \min_{[\Delta \tilde{t}_{12}, \Delta \tilde{t}_{13}] \in \mathcal{T}_\textrm{r,s}} P_\textrm{e,r}.
\end{aligned}
\end{equation}
where 
$\mathcal{T}_\textrm{r, s}$ is characterized in Lemma~\ref{lemma2}. The values of $t_{i}$s and $\tilde{t}_j$s can be obtained from $\Delta \tilde{t}_{12,\textrm{r,s}}, \Delta \tilde{t}_{13,\textrm{r,s}}$,
which are noted by $\tilde{t}_{1,\textrm{r,s}}$, $\tilde{t}_{2,\textrm{r,s}}$, $\tilde{t}_{3,\textrm{r,s}}$, $t_{1,\textrm{r,s}}$, $t_{2,\textrm{r,s}}$, and $t_{3,\textrm{r,s}}$.
 
The optimum sampling times in each case are obtained numerically in Section \ref{simulation}. We note that the values of the channel coefficients and hence the error probability depend on $\Delta t_{ij}^{[l]}$s. Hence, if we shift the releasing and sampling times by a same value, the error probability does not change and thus the optimum times are not unique. We choose the times that are all positive and the smallest one is zero.

\section{Error Performance Analysis of the Proposed IA Scheme}\label{sec:error_prob}
In the proposed IA scheme, according to \eqref{received_signals}, the mean signal vector
 at $\textrm{Rx}_i$, assuming no reaction is
\begin{align}\label{IA_eq}
\bm{\mu}_i=\tilde{\bm{H}}_{i}N_i+\tilde{\bm{H}}_{\textrm{I},i}N_{\textrm{I},i},\qquad i \in \{1,2,3\},
\end{align}
where 
$N_{\textrm{I},1}=\frac{N_2}{A}+\frac{N_3}{B}$, $N_{\textrm{I},2}=\frac{N_1}{2}+\frac{N_3}{B}$, and $N_{\textrm{I},3}=\frac{N_1}{2}+\frac{N_2}{A}$; $A$ and $B$ are defined in \eqref{received_signals}. Further, 
 $\tilde{\bm{H}}_{i}=[\tilde{H}_{i}^{[1]},\tilde{H}_{i}^{[2]}]^\textrm{T}$ and $\tilde{\bm{H}}_{\textrm{I},i}=[\tilde{H}_{\textrm{I},i}^{[1]},\tilde{H}_{\textrm{I},i}^{[2]}]^\textrm{T}$ can be obtained from the channel coefficients using \eqref{received_signals}.

\textbf{IA without reaction}: The message of $\textrm{Tx}_i$ ($M_i$) can be obtained from the number of released molecules, $N_i$.
Using \eqref{IA_eq}, $N_i$ can be obtained at $\textrm{Rx}_i$ from the mean number of received molecules  by zero-forcing as $N_i=\frac{\tilde{H}_{\textrm{I},i}^{[2]}\mu_i^{[1]}-\tilde{H}_{\textrm{I},i}^{[1]}\mu_i^{[2]}}{\tilde{H}_{i}^{[1]}\tilde{H}_{\textrm{I},i}^{[2]}-\tilde{H}_{i}^{[2]}\tilde{H}_{\textrm{I},i}^{[1]}}$. Due to the counting noise, the number of received molecules of type $l$, $Y_i^{[l]}$, has Poisson distribution with mean $\mu_i^{[l]}$. Using the noisy observations, we obtain a random variable $\tilde{N}_i=\frac{\tilde{H}_{\textrm{I},i}^{[2]}Y_i^{[1]}-\tilde{H}_{i}^{[1]}Y_i^{[2]}}{\tilde{H}_{i}^{[1]}\tilde{H}_{\textrm{I},i}^{[2]}-\tilde{H}_{i}^{[2]}\tilde{H}_{\textrm{I},i}^{[1]}}$, with mean $N_i$, whose distribution can be obtained from the distribution of $Y_i^{[1]}$ and $Y_i^{[2]}$. We use maximum a-posteriori (MAP) decision rule on $\tilde{N}_i$ to obtain $M_i$:
\begin{align}
\Prob(\tilde{N}_i=\tilde{n}_i|M_i=1)\overset{M_i=1}{\underset{M_i=0}\gtrless}\Prob(\tilde{N}_i=\tilde{n}_i|M_i=0).
\end{align}
For $\textrm{Rx}_1$, the above equation results in
\begin{align}\label{dec_rule}
\nonumber
&\sum_{m_2,m_3 \in \{0,1\}}\Prob(\tilde{n}_1|M_1=1,M_2=m_2,M_3=m_3)\overset{M_1=1}{\underset{M_1=0}\gtrless}\\&\sum_{m_2,m_3 \in \{0,1\}}\Prob(\tilde{n}_1|M_1=0,M_2=m_2,M_3=m_3).
\end{align}
To obtain the MAP decision rule, we need to obtain the distribution of $\tilde{N}_i$ conditioned on the messages of the transmitters. $\tilde{N}_i$ is a linear combination of two Poisson random variables $Y_i^{[1]}$ and $Y_i^{[2]}$, i.e., $\tilde{N}_i=a_iY_i^{[1]}+b_iY_i^{[2]}$, where $a_i=\frac{\tilde{H}_{\textrm{I},i}^{[2]}}{\tilde{H}_{i}^{[1]}\tilde{H}_{\textrm{I},i}^{[2]}-\tilde{H}_{i}^{[2]}\tilde{H}_{\textrm{I},i}^{[1]}}$ and $b_i=\frac{-\tilde{H}_{\textrm{I},i}^{[1]}}{\tilde{H}_{i}^{[1]}\tilde{H}_{\textrm{I},i}^{[2]}-\tilde{H}_{i}^{[2]}\tilde{H}_{\textrm{I},i}^{[1]}}$. This yields to a very complex problem. For simplicity, we assume Gaussian approximation to the Poisson distribution of $Y_i^{[l]}$, i.e., ${Y_i^{[l]}}_{|m_1,m_2,m_3} \sim \mathcal{N}(\mu_i^{[l]}, \sqrt{\mu_i^{[l]}})$, which is a good approximation if the number of released molecules is large. Since different molecule types do not interfere with each other, $Y_i^{[1]}$ and $Y_i^{[2]}$ are independent conditioned on the messages of the transmitters. Hence, $\tilde{N}_i$ has a Gaussian distribution conditioned on $M_1=m_1$, $M_2=m_2$, and $M_3=m_3$ with mean $\tilde{\mu}_{i}(m_1,m_2,m_3)=a_i \mu_i^{[1]}+b_i\mu_i^{[2]}$ and variance $\tilde{\sigma}_{i}^2(m_1,m_2,m_3)=a_i^2\mu_i^{[1]}+b_i^2\mu_i^{[2]}$. Therefore, the MAP decision rule in \eqref{dec_rule} becomes
\small
\begin{align}\label{dec_rule2}
&\sum_{m_2,m_3 \in \{0,1\}}\big[\frac{1}{\sqrt{2 \pi \tilde{\sigma}_{i}^2(1,m_2,m_3)}}e^{-\frac{(\tilde{n}_1-\tilde{\mu}_{i}(1,m_2,m_3))^2}{2\tilde{\sigma}_{i}^2(1,m_2,m_3)}}\\\nonumber
& \qquad \qquad-\frac{1}{\sqrt{2 \pi \tilde{\sigma}_{i}^2(0,m_2,m_3)}}e^{-\frac{(\tilde{n}_1-\tilde{\mu}_{i}(0,m_2,m_3))^2}{2\tilde{\sigma}_{i}^2(0,m_2,m_3)}}\big]\overset{M_1=1}{\underset{M_1=0}\gtrless} 0.
\end{align}
\normalsize
However, obtaining a closed-form of this decision rule is not possible and we derive it numerically in Section \ref{simulation}. Since the MAP decision rule does not yield to a simple form, we cannot obtain a closed-form equation for the error probability of this scheme and we obtain the error probability using simulation in Section \ref{simulation}.

\textbf{IA with reaction}: For the proposed IA scheme with perfect reaction, the average signal vector at $\textrm{Rx}_i$ after reaction is, \cite{farahnak2018medium},
\small
\begin{align*}
\bm{\mu}_{\textrm{r},i}
=[\mu_{\textrm{r},i}^{[1]},\mu_{\textrm{r},i}^{[2]}]^\textrm{T}=\begin{cases}
\big[(\tilde{H}_{i}^{[1]}-c_i\tilde{H}_{i}^{[2]})N_i,0\big]^\textrm{T}, &\textrm{if}~\frac{\tilde{H}_{i}^{[1]}}{\tilde{H}_{i}^{[2]}}>c_i\\
\big[0,(\tilde{H}_{i}^{[2]}-\frac{1}{c_i}\tilde{H}_{i}^{[1]})N_i\big]^\textrm{T}, &\textrm{if}~\frac{\tilde{H}_{i}^{[1]}}{\tilde{H}_{i}^{[2]}}<c_i
\end{cases}.
\end{align*}
\normalsize
Again, the number of received molecules of type $l$ after reaction, $Y_{\textrm{r},i}^{[l]}$, has Poisson distribution with mean $\mu_{\textrm{r},i}^{[l]}$ \cite{farahnak2018medium}. 
 If $\frac{\tilde{H}_{i}^{[1]}}{\tilde{H}_{i}^{[2]}}>c_i$, we have $y_{\textrm{r},i}^{[2]}=0$, and hence, we use $y_{\textrm{r},i}^{[1]}$ to detect $M_i$ using MAP decision rule, $\Prob(y_{\textrm{r},i}^{[1]}|M_i=1)\overset{M_i=1}{\underset{M_i=0}\gtrless}\Prob(y_{\textrm{r},i}^{[1]}|M_i=0),$
which results in a simple threshold rule as $y_{\textrm{r},i}^{[1]}\overset{M_i=1}{\underset{M_i=0}\gtrless}\gamma_i^{[1]}$, where $\gamma_i^{[1]}=\frac{(\tilde{H}_{i}^{[1]}-c_i\tilde{H}_{i}^{[2]})(\zeta_1-\zeta_0)}{\ln(\frac{\zeta_1}{\zeta_0})}$. Hence, for $\frac{\tilde{H}_{i}^{[1]}}{\tilde{H}_{i}^{[2]}}>c_i$, the error probability at $\textrm{Rx}_i$, $i \in \{1,2,3\}$, is obtained as
\small
\begin{align}\label{eq_Pe_IA}
\nonumber
P_{\textrm{e},i}&=\frac{1}{2}(P\{Y_{\textrm{r},i}^{[1]}>\gamma_i^{[1]}|M_i=0\}+P\{Y_{\textrm{r},i}^{[1]}<\gamma_i^{[1]}|M_i=1\})\\\nonumber
&=\frac{1}{2}\big[1-\sum_{y_{\textrm{r},i}^{[1]}=0}^{\lfloor \gamma_i^{[1]}\rfloor}\big(\frac{((\tilde{H}_{i}^{[1]}-c_i\tilde{H}_{i}^{[2]})\zeta_0)^{y_{\textrm{r},i}^{[1]}} e^{-(\tilde{H}_{i}^{[1]}-c_i\tilde{H}_{i}^{[2]})\zeta_0}}{y_{\textrm{r},i}^{[1]}!}\\
&\quad-\frac{((\tilde{H}_{i}^{[1]}-c_i\tilde{H}_{i}^{[2]})\zeta_1)^{y_{\textrm{r},i}^{[1]}} e^{-(\tilde{H}_{i}^{[1]}-c_i\tilde{H}_{i}^{[2]})\zeta_1}}{y_{\textrm{r},i}^{[1]}!}\big)\big].
\end{align}
\normalsize
The error probability for the case of $\frac{\tilde{H}_{i}^{[1]}}{\tilde{H}_{i}^{[2]}}<c_i$ can be obtained similarly.

\section{Simulation and Numerical Results}\label{simulation}
We evaluate the performance of the proposed IA scheme in a 3-user IFC. We consider the parameters in Table~\ref{parameters}.

\renewcommand{\arraystretch}{0.9}
\begin{table}
\centering
\captionof{table}{Simulation and numerical analysis parameters}
\begin{tabular}{p{3.7cm}|l|l}
\hline\hline
Parameter & Symbol & Value\\\hline
Location of Tx$_1$ & $r_{1}^{\textrm{Tx}}$ & $(0,0,0)~\mu \textrm{m}$ \\
Location of Tx$_2$ &$r_{2}^{\textrm{Tx}}$ & $(0,20,10)~\mu \textrm{m}$\\
Location of Tx$_3$ &$r_{3}^{\textrm{Tx}}$ & $(0,0,30)~\mu \textrm{m}$\\
Location of Rx$_1$ &$r_{1}^{\textrm{Rx}}$ & $(0,150,0)~\mu \textrm{m}$\\
Location of Rx$_2$ &$r_{2}^{\textrm{Rx}}$ & $(0,200,10)~\mu \textrm{m}$\\
Location of Rx$_3$ &$r_{3}^{\textrm{Rx}}$ & $(0,300,20)~\mu \textrm{m}$\\
Diffusion coefficient of M$_1$ & $D_1$ & $10^{-8}~\textrm{m}^2/\textrm{s}$\\
Diffusion coefficient of M$_2$ & $D_2$ & $5 \times 10^{-8}~\textrm{m}^2/\textrm{s}$\\
Reaction coefficients at Rx$_1$, Rx$_2$, and Rx$_3$ & $c_1$, $c_2$, $c_3$ & $2$\\
Receiver radius & $r_\textrm{R}$ & $15~\mu \textrm{m}$\\
Medium flow velocity & $\bm{\nu}$ & $\bm{0}$\\
\hline\hline
\end{tabular}
\label{parameters}
\vspace{-1.5em}
\end{table}

\begin{figure*}
\centering
~~\begin{minipage}{0.31\textwidth}
\centering
\includegraphics[trim={3cm 0 0 0},scale=0.34]{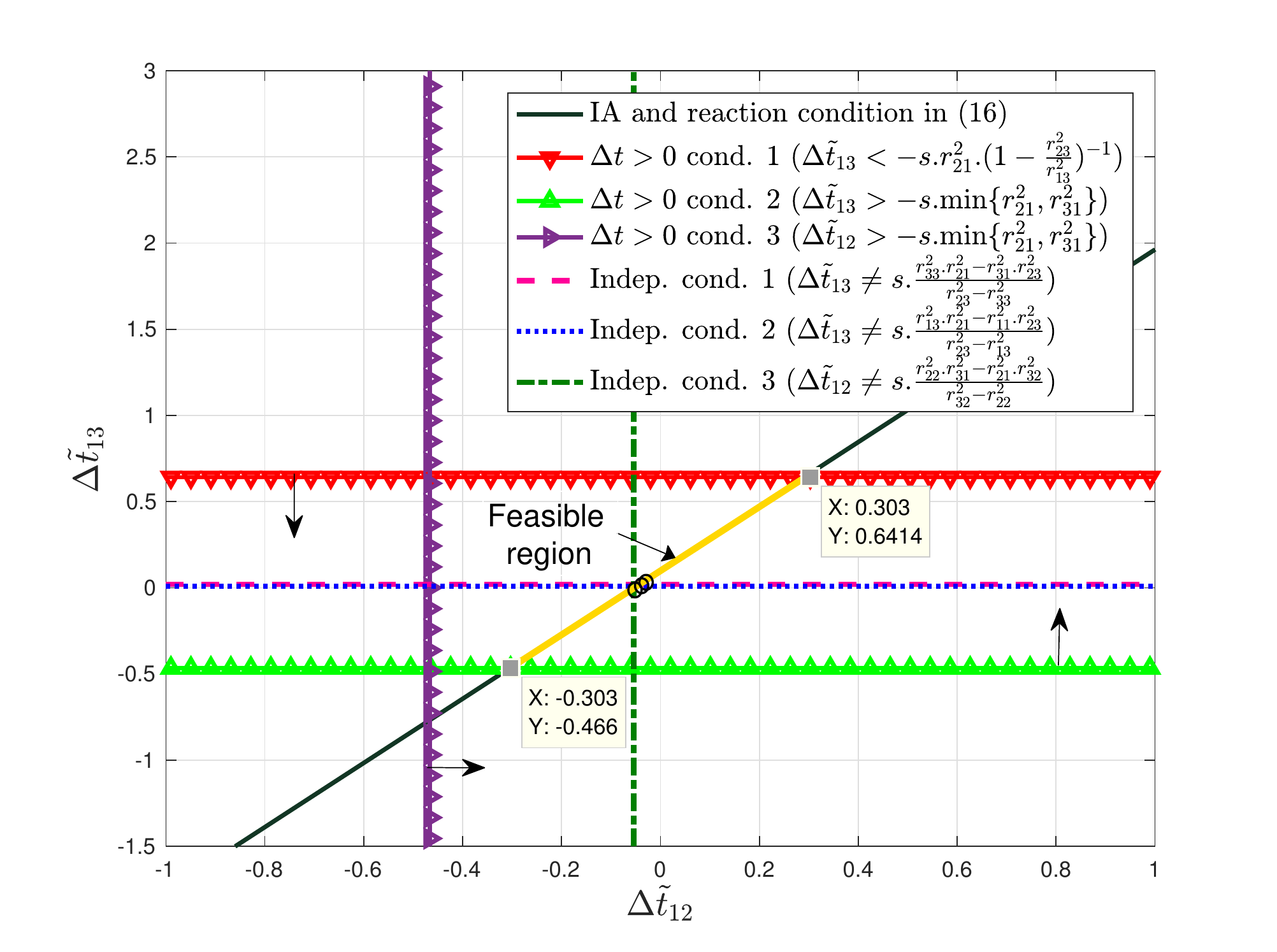}
\caption{Feasible region of $\Delta \tilde{t}_{12}$ and $\Delta \tilde{t}_{13}$ in the proposed IA scheme with reaction.}
\label{figfeasible}
\end{minipage}\quad
\begin{minipage}{0.31\textwidth}
\centering
\includegraphics[trim={1.9cm 0 0 0},scale=0.34]{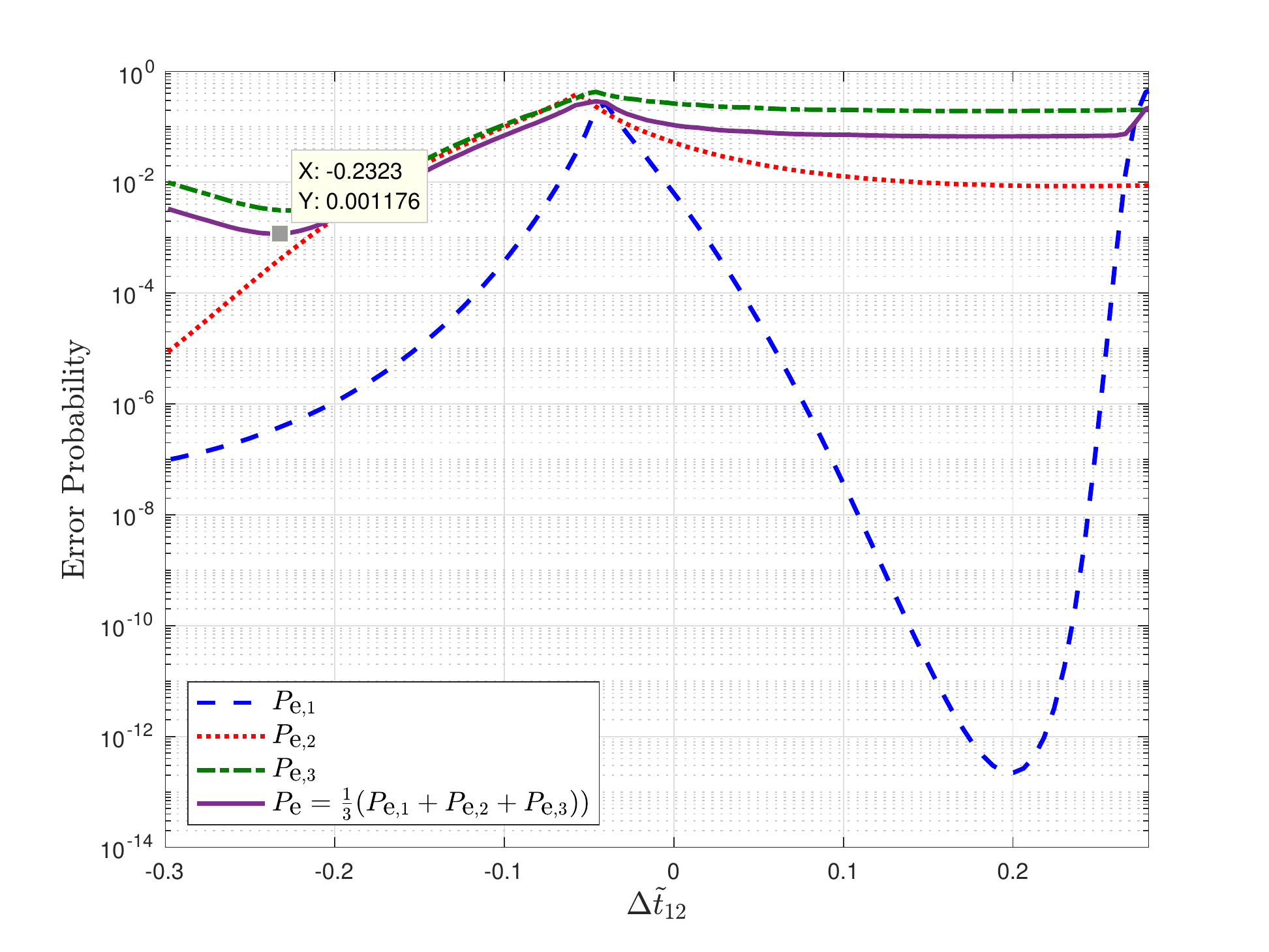}
\caption{The error probability versus $\Delta t_{12}$ for the proposed IA scheme with reaction.}
\label{figoptdtr}
\end{minipage}\quad
\begin{minipage}{0.31\textwidth}
\centering
\includegraphics[trim={1.4cm 0 0 0},scale=0.34]{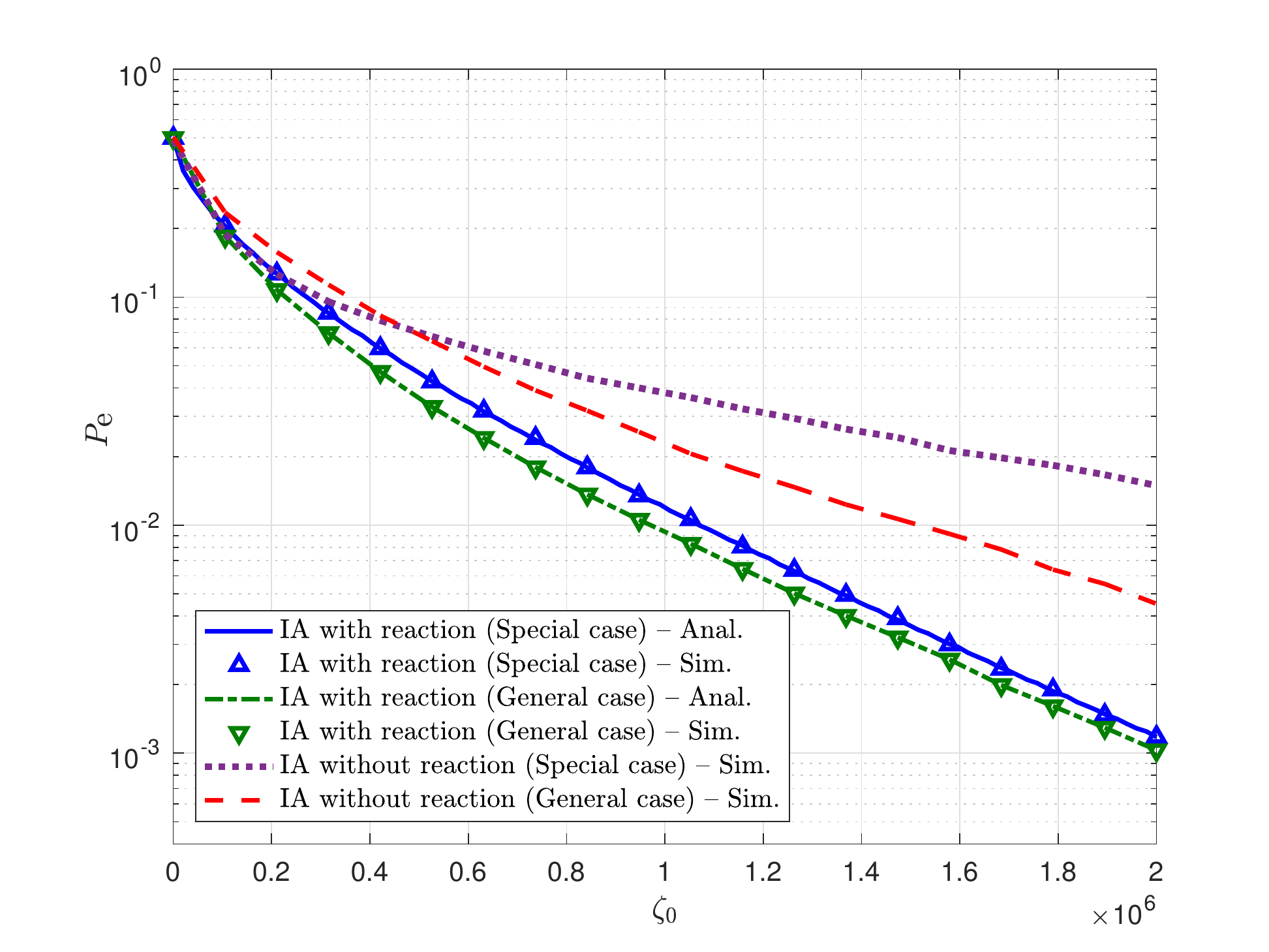}
\caption{The total error probability versus $\zeta_0$ for the proposed IA scheme with and without reaction.}
\label{figoverallPe}
\end{minipage}
\label{figtotalestimator}
\vspace{-0.5em}
\end{figure*}




The feasible region in the proposed IA scheme with reaction for $\Delta \tilde{t}_{12}$ and $\Delta \tilde{t}_{13}$ for the special case of the releasing and sampling times ($t_{i}^{[1]}=t_{i}^{[2]}=t_{i}, \tilde{t}_{j}^{[1]}=\tilde{t}_{j}^{[2]}=\tilde{t}_{j}$, $i,j \in \{1,2,3\}$), i.e., $\mathcal{T}_\textrm{r, s}$, based on Lemma~\ref{lemma2}, is shown in Fig.~\ref{figfeasible}, which is a line segment between the points $(-0.303, -0.466)$ and $(0.303, 0.6414)$, where three points $(-0.0538,0.0005)$, $(-0.0492,0.0092)$, and $(-0.0434,0.0201)$ are excluded from it.  
The error probability at the receivers and the total error probability of the system versus $\Delta \tilde{t}_{12}$ in the feasible region is shown in Fig.~\ref{figoptdtr}.\footnote{Note that since the sampling times are not necessarily chosen in the peak times of the received concentrations, $P_\textrm{e}$ does not go below $0.0012$, while $P_{\textrm{e},1}$ is near $10^{-14}$ for some sampling and releasing times. However, if we do not use IA, for 3 users with 2 molecules types, two transmitters should use one molecule type and the error probability of them becomes near $0.5$ because of the high co-channel interference, and hence 
the total error probability of system will be $P_\textrm{e} \approx 0.5$, 
which is not acceptable.} Here, we assume $\zeta_0=\frac{\zeta_1}{2}=2\times 10^{6}$. It can be seen that the optimum value of $\Delta \tilde{t}_{12}$ which minimizes the total error probability is $\Delta \tilde{t}_{12,\textrm{r,s}}=-0.2323$, and using \eqref{IA_reaction7}, we have $\Delta \tilde{t}_{13,\textrm{r,s}}=-0.3317$. Now, assuming $\tilde{t}_{1,\textrm{r,s}}=0$, we obtain the other releasing times ($\tilde{t}_{2,\textrm{r,s}}$ and $\tilde{t}_{3,\textrm{r,s}}$), and then from \eqref{IA_reaction9}, we obtain $t_{i,\textrm{r,s}}$s, which are provided in Table \ref{table_times}.\footnote{Since $\Delta \tilde{t}_{12,\textrm{r,s}}$ and $\Delta \tilde{t}_{13,\textrm{r,s}}$ are both negative, by assuming $\tilde{t}_{1,\textrm{r,s}}=0$, we can assure to have positive times with smallest one equal to zero.}
 For the general case of the times in the IA scheme with reaction, the optimum times are obtained numerically from \eqref{times_ar_gen} in Table \ref{table_times}.  
The optimum sampling and releasing times in the IA scheme without reaction are also provided in Table \ref{table_times}. 
The optimum times are obtained numerically from \eqref{times_anr_spec} for the special case, and from \eqref{times_anr_gen} for the general case,


 \begin{table*}
\centering
\captionof{table}{Optimum releasing and sampling times}
\begin{tabular}{p{4.4cm}|p{0.65cm}|p{0.65cm}|p{0.65cm}|p{0.65cm}|p{0.65cm}|p{0.65cm}|p{0.65cm}|p{0.65cm}|p{0.65cm}|p{0.65cm}|p{0.65cm}|p{0.65cm}}
\hline\hline
 & $\tilde{t}_{1}^{[1]}$ & $\tilde{t}_{1}^{[2]}$ & $\tilde{t}_{2}^{[1]}$ & $\tilde{t}_{2}^{[2]}$ & $\tilde{t}_{3}^{[1]}$ & $\tilde{t}_{3}^{[2]}$ & $t_{1}^{[1]}$ & $t_{1}^{[2]}$ & $t_{2}^{[1]}$ & $t_{2}^{[2]}$ & $t_{3}^{[1]}$ & $t_{3}^{[2]}$\\\hline
IA with reaction (Special case) & $0$ & $0$ & $0.232$ & $0.232$ & $0.332$ & $0.332$ & $0.410$ & $0.410$ & $0.466$ & $0.466$ & $1.051$ & $1.051$\\
IA with reaction (General case)  &
 $ 0.012$ & $0$ & $0.335$ & $0.396$ & $0.329$ & $0.323$ & $0.411$ & $0.411$ & $0.471$ & $0.471$ & $1.055$ & $1.055$\\
IA without reaction (Special case) & $0.592$ & $0.592$ & $0$ & $0$ & $0.623$ & $0.623$ & $0.674$ & $0.674$ & $0.666$ & $0.666$ & $1.758$ & $1.758$\\
IA without reaction (General case) & $1.262$ & $0.004$ & $0$ & $1.079$ & $1.102$ & $1.555$ & $1.636$ & $1.575$ & $1.392$ & $2.385$ & $1.877$ & $1.814$\\
\hline\hline
\end{tabular}
\label{table_times}
\vspace{-2em}
\end{table*}

The total error probabilities of the system using the IA scheme with and without reaction for the optimum times in the general and special cases are shown in Fig.~\ref{figoverallPe} versus the number of released molecules for bit 0 (i.e., $\zeta_0$). Here, we assume $\zeta_1=2\zeta_0$. For the IA with reaction, we obtain the error probability using the analytical result in \eqref{eq_Pe_IA} and simulation, and for the IA without reaction, we obtain the decision rule numerically from \eqref{dec_rule2} and obtain the error probability using simulation.
It is seen that using reaction in the IA scheme improves the total error probability significantly.
Further, for the IA scheme with reaction, the simulations and the analysis provide the same result.

\begin{figure}[t]
\centering
\includegraphics[trim={0cm 0 0 0cm}, scale=0.34]{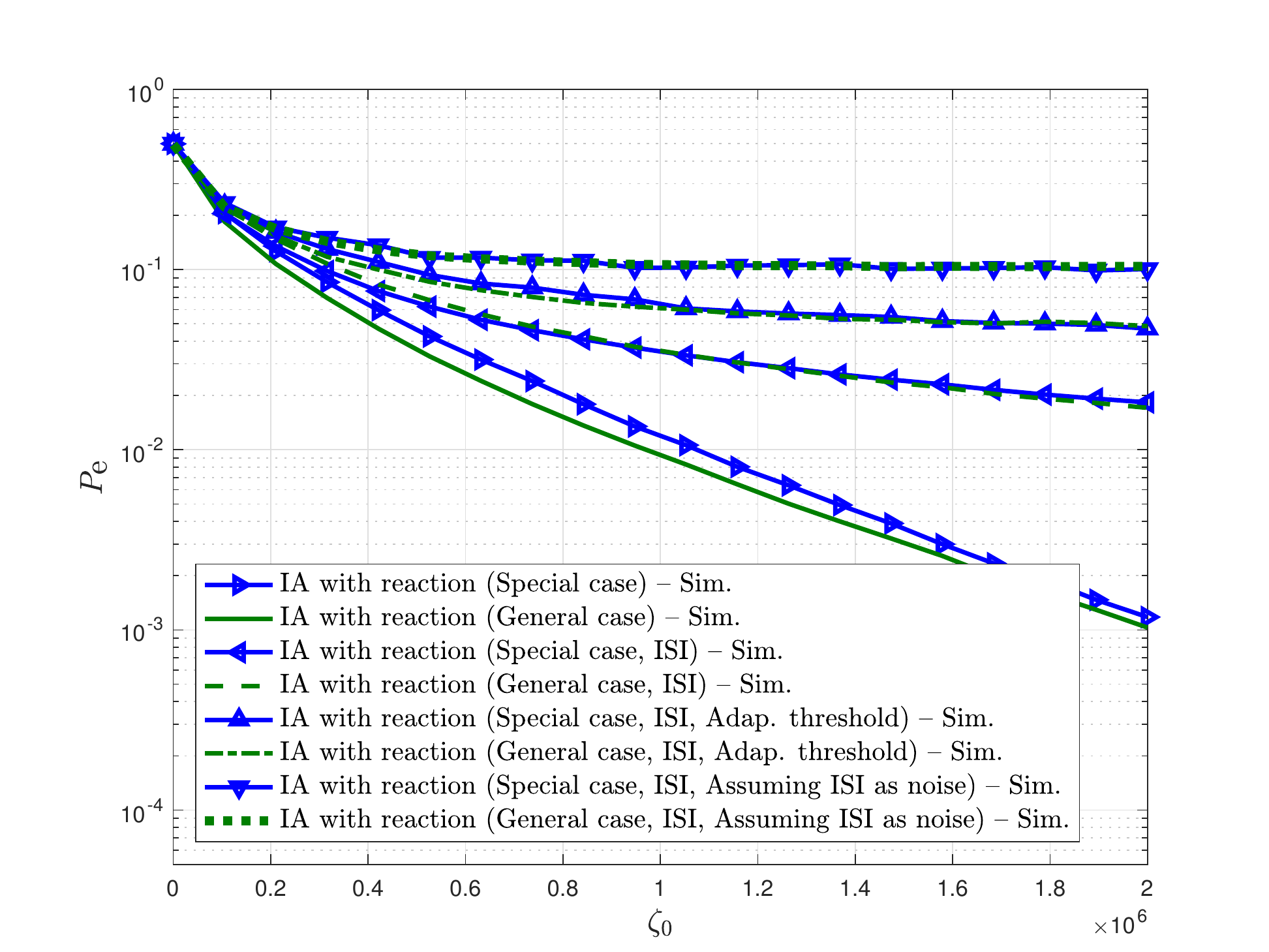}
\vspace{-0.5em}
\caption{The total error probability versus $\zeta_0$ for the proposed IA scheme with reaction in the presence of ISI.}
\label{figoverallPe_ISI}
\vspace{-1.5em}
\end{figure}

\begin{figure}[t]
\centering
\includegraphics[trim={0cm 0 0 0}, scale=0.34]{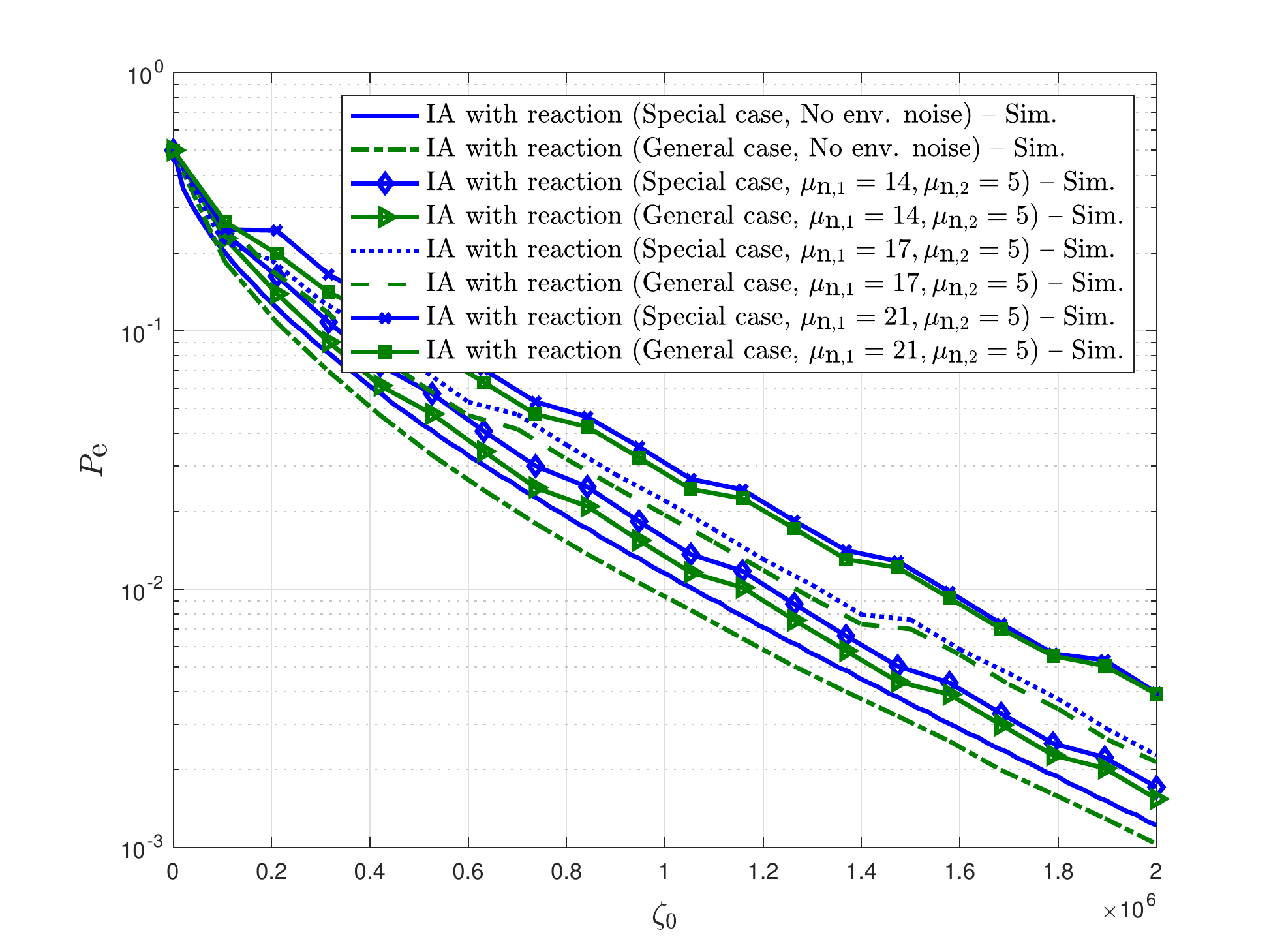}
\vspace{-0.5em}
\caption{The total error probability versus $\zeta_0$ for the proposed IA scheme with reaction in the presence of the environment noise.}
\label{figoverallPe_noise}
\vspace{-1.5em}
\end{figure}

Fig.~\ref{figoverallPe_ISI} investigates the ISI effect in the system. In this figure, the total error probability of the system using IA with reaction assuming one slot memory in the channels along with the error probability in the no ISI case are shown versus $\zeta_0$. We obtain the error probabilities for the optimum decision rule, the sub-optimum decision rule assuming ISI as noise, and the adaptive decision rule using the previously detected bits introduced in \cite{mosayebi2014}. Here, we assume $T_\textrm{s}=10~ \textrm{s}$. It can be seen that the error performance degrades in the presence of ISI. However, we have an acceptable performance using adaptive and optimum decision rules. \textcolor{black}{The effect of the environment noise on the proposed IA scheme with reaction is shown in Fig.~\ref{figoverallPe_noise} for $\mu_{n}^{[1]}=14, 17, 21$ and $\mu_{n}^{[2]}=5$, which shows that the proposed scheme performs well in the presence of noise. It should be noted that, due to the perfect reaction between the molecules of type 1 and 2 with coefficient $c$ at the receivers, the average noise concentration for molecule type 1 is $\textrm{max}\{0,\mu_{\textrm{n},1}-c\mu_{\textrm{n},2}\}$, and the average noise concentration for molecule type 2 is $\textrm{max}\{0,\mu_{\textrm{n},2}-\mu_{\textrm{n},1}/c\}$. Therefore, the error probability depends on the value of $\mu_{\textrm{n},1}-c\mu_{\textrm{n},2}$ and increases as $\mu_{\textrm{n},1}-c\mu_{\textrm{n},2}$ goes away from zero. }

%
%

\section{Conclusion}\label{conclusion}
In this paper, we proposed an interference alignment (IA) scheme for a 3-user molecular interference channel (IFC) based on the choice of the releasing times of the molecules at the transmitters and the sampling times at the receivers using two molecule types. We proposed using reaction in the IA scheme to cancel the aligned interference signal in the medium and reduce the signal dependent noise. Further, we applied the asymptotic IA scheme in classic communications to a $K$-user molecular IFC. For the proposed IA scheme, we obtained the feasible region for the releasing/sampling times and simplified it for a special case. We investigated the error performance of the IFC using the proposed IA scheme with and without reaction and showed that the proposed IA scheme using reaction improves the performance significantly.

\bibliographystyle{ieeetr}
\bibliography{reftest}

\appendices
\section{Proof of Lemma~\ref{lemma2}}\label{apendixa}
For $t_{i}^{[1]}=t_{i}^{[2]}=t_{i}, \tilde{t}_{j}^{[1]}=\tilde{t}_{j}^{[2]}=\tilde{t}_{j}$, $i,j \in \{1,2,3\}$, $\bm{\nu}=0$, and $c_1=c_2=c_3=c$, the reaction conditions in \eqref{IA_reaction2} simplify to
\begin{align}\label{IA_reaction_simplified}
\textrm{Rx}_1: &\frac{r_{12}^2}{t_1-\tilde{t}_2}+\frac{r_{31}^2}{t_{3}-\tilde{t}_1}-\frac{r_{32}^2}{t_3-\tilde{t}_2}=\frac{1}{s},\\\nonumber
\textrm{Rx}_2: &\frac{r_{21}^2}{t_2-\tilde{t}_1}=\frac{1}{s},\quad \textrm{Rx}_3: \frac{r_{31}^2}{t_3-\tilde{t}_1}=\frac{1}{s},
\end{align}
where $s=-\frac{\frac{1}{D_1}-\frac{1}{D_2}}{4\ln{c}+6\ln(\frac{D_1}{D_2})}$. Combining the IA condition in \eqref{cond1_simplified2} and the reaction conditions in \eqref{IA_reaction_simplified} results in
\begin{subequations}
\label{IA_reac_simplified}
\begin{align}\label{condIA_reac_simplifieda}
& r_{12}^2.(t_3-\tilde{t}_2)=r_{32}^2.(t_1-\tilde{t}_2), \\\label{condIA_reac_simplifiedb}
&r_{13}^2.(t_{2}-\tilde{t}_{3})=r_{23}^2.(t_{1}-\tilde{t}_3),\\\label{condIA_reac_simplifiedc}
&t_2-\tilde{t}_1=s.r_{21}^2,\qquad t_3-\tilde{t}_1=s.r_{31}^2.
\end{align}
\end{subequations}
From \eqref{condIA_reac_simplifiedb}-\eqref{condIA_reac_simplifiedc}, $t_1$, $t_2$, and $t_3$ can be obtained from $\tilde{t}_1$ and $\tilde{t}_3$ as in \eqref{IA_reaction9} and if we substitute them in \eqref{condIA_reac_simplifieda}, we obtain \eqref{IA_reaction7}. 
Combining the independency conditions in \eqref{cond2_simplified2} and the reaction conditions in \eqref{IA_reaction_simplified} results in
\begin{align}\label{cond2_simplified5}
\nonumber
&\textrm{Rx}_1:t_{1}-\tilde{t}_{1} \neq s.r_{11}^2, \textrm{Rx}_2:r_{32}^2.(t_{2}-\tilde{t}_{2}) \neq r_{22}^2.(t_{3}-\tilde{t}_{2}),\\
&\textrm{Rx}_3:r_{33}^2.(t_{2}-\tilde{t}_{3}) \neq r_{23}^2.(t_{3}-\tilde{t}_{3}).
\end{align}
Now, substituting $t_i$s from \eqref{IA_reaction9}, we obtain \eqref{IA_reaction8}.\\
Further, we should have $\Delta t_{ij}=t_i-\tilde{t}_j>0$, $i,j \in \{1,2,3\}$. From \eqref{IA_reaction9}, the values of $\Delta t_{ij}$ are related to $\Delta \tilde{t}_{12}$ and $\Delta \tilde{t}_{13}$ as follows:
\begin{align}\label{deltats}
\nonumber
&\Delta t_{11}=t_{1}-\tilde{t}_{1}=(\frac{r_{13}^2}{r_{23}^2}-1)\Delta \tilde{t}_{13}+s.\frac{r_{13}^2}{r_{23}^2}.r_{21}^2,\\\nonumber
&\Delta t_{12}=t_{1}-\tilde{t}_{2}=\Delta \tilde{t}_{12}+(\frac{r_{13}^2}{r_{23}^2}-1)\Delta \tilde{t}_{13}+s.\frac{r_{13}^2}{r_{23}^2}.r_{21}^2,\\\nonumber
&\Delta t_{13}=t_{1}-\tilde{t}_{3}=\frac{r_{13}^2}{r_{23}^2}\Delta \tilde{t}_{13}+s.\frac{r_{13}^2}{r_{23}^2}.r_{21}^2,\\\nonumber
&\Delta t_{21}=t_{2}-\tilde{t}_{1}=s.r_{21}^2, \Delta t_{22}=t_{2}-\tilde{t}_{2}=\Delta \tilde{t}_{12}+s.r_{21}^2,\\\nonumber
&\Delta t_{23}=t_{2}-\tilde{t}_{3}=\Delta \tilde{t}_{13}+s.r_{21}^2,\\\nonumber
&\Delta t_{31}=t_{3}-\tilde{t}_{1}=s.r_{31}^2, \Delta t_{32}=t_{3}-\tilde{t}_{2}=\Delta \tilde{t}_{12}+s.r_{31}^2,\\
&\Delta t_{33}=t_{3}-\tilde{t}_{3}=\Delta \tilde{t}_{12}+s.r_{31}^2.
\end{align}
Hence, $\Delta t_{ij}>0$ and \eqref{deltats} result in \eqref{IA_reaction10}.


\end{document}